\documentclass[aps,pra,twocolumn,showpacs,superscriptaddress]{revtex4}

\usepackage{amsfonts,epsfig}
\usepackage{latexsym}
\usepackage{epsfig}
\usepackage{amsmath}
\usepackage{amssymb}
\usepackage{amsfonts}
\usepackage{amsthm}
\usepackage{mathrsfs}
\usepackage{natbib}
\usepackage{color,verbatim,graphics}
\DeclareMathAlphabet{\mathrsfs}{U}{rsfs}{m}{n}
\DeclareMathAlphabet{\mathpzc}{OT1}{pzc}{m}{it}
\DeclareMathAlphabet{\matheus}{U}{eus}{m}{n}
\DeclareMathAlphabet{\mathbbold}{U}{bbold}{m}{n}

\newtheorem{theorem}{Theorem}

\newtheorem{lemma}[theorem]{Lemma}

\newcommand{\ba}{\begin{eqnarray}}
\newcommand{\ea}{\end{eqnarray}}
\newcommand{\ban}{\begin{eqnarray*}}
\newcommand{\ean}{\end{eqnarray*}}
\newcommand{\Tr}{\operatorname{Tr}}

\def\be{\begin{equation}}
\def\ee{\end{equation}}
 
\newcommand{\ave}[1]{\langle#1\rangle}
\newcommand{\ket}[1]{|#1\rangle}

\newcommand{\sgn}{\operatorname{sgn}}

\newcommand{\referfig}[1]{Figure \ref{#1}}

\newcommand{\matrixww}[4]{\left( \begin{array}{cc} #1 & #2 \\#3 & #4 \\ \end{array} \right) }

\begin{document}

\title{Quantum Bell Inequalities from Macroscopic Locality}

\author{Tzyh Haur Yang}
\affiliation{Centre for Quantum Technologies, National University of Singapore, Singapore}
\author{Miguel Navascu\'es}
\affiliation{Depto An\'alisis Matem\'atico and IMI, Universidad Complutense de Madrid}
\author{Lana Sheridan}
\affiliation{Centre for Quantum Technologies, National University of Singapore, Singapore}
\author{Valerio Scarani}
\affiliation{Centre for Quantum Technologies, National University of Singapore, Singapore}
\affiliation{Department of Physics, National University of Singapore, Singapore}

\date{\today}


\begin{abstract}
We propose a method to generate analytical quantum Bell inequalities based on the principle of Macroscopic Locality. By imposing locality over binary processings of virtual macroscopic intensities, we establish a correspondence between Bell inequalities and quantum Bell inequalities in bipartite scenarios with dichotomic observables. We discuss how to improve the latter approximation and how to extend our ideas to scenarios with more than two outcomes per setting.
\end{abstract}


\maketitle


\section{Introduction}
Any probability distribution achievable in the measurement of a bipartite quantum system can be written as
\ban
P(a,b|x,y)&=&\Tr\left(\rho\,\Pi_a^x\otimes \Pi_b^y\right)
\label{tensor}
\ean
where $\rho$ is a density operator of unrestricted dimension and the $\Pi$'s are projectors. The set of quantum correlations (or \textit{quantum set}) is the set of all probability distributions that can be written in this form. It has long been known that some correlations in this set are \emph{nonlocal}, in the sense that they cannot be produced by pre-established agreement \cite{Bell1966}. On another front, the quantum set is strictly smaller than the \emph{no-signaling set}, i.e. those correlations that cannot be used to send a signal \cite{Popescu1994}.

The mathematical characterization of the quantum set is not trivial \cite{Cirel'son1980}. One of the few definite results was obtained for the case of two dichotomic measurements per site, i.e., $x$, $y$, $a$ and $b$ can take only two values each (a situation that we shall refer to as $2222$). In this scenario, up to symmetries, the only Bell inequality is the Clauser-Horne-Shimony-Holt (CHSH) \cite{Clauser1969}
\ba
\sum_{x,y} (-1)^{xy}\,\ave{a_xb_y}&\leq& 2,
\ea where for convenience of notation we have chosen $x,y\in\{0,1\}$, $a,b\in\{-1,+1\}$. Tsirelson and Landau (and, much later, Masanes) independently proved that the set of quantum correlations is bounded by a \textit{quantum Bell inequality} obtained from CHSH upon replacing $\ave{a_xb_y}$ by $\frac{2}{\pi}\arcsin\ave{a_xb_y}$ \cite{Tsirelson1987,Landau1988,Masanes2003}. More recently, Navascu\'es, Pironio and Ac\'{\i}n proved that, if we relax the tensor structure assumption in (\ref{tensor}) and instead demand commutativity between Alice's and Bob's measurement operators, the resulting set of correlations can be approximated, and ultimately reached, by a hierarchy of semi-definite programs \cite{Navascues2007,Navascues2008}. In the process, they proved that the replacement $\ave{a_xb_y} \longrightarrow \frac{2}{\pi}\arcsin D_{xy}$, where
\ba
D_{xy}=\frac{\ave{a_xb_y}-\ave{a_x}\ave{a_y}}{\sqrt{(1-\ave{a_x}^2)(1-\ave{b_y}^2)}} \label{replace}
\ea in the CHSH inequality leads to a quantum Bell inequality that is stronger than the previously known one and identifies the set denoted $Q^1$, the first step of the semi-definite hierarchy.

Aside from from the mathematical line of attack, several authors have followed the spirit of the Popescu-Rohrlich paper \cite{Popescu1994} and tried to identify \textit{physical principles} that would define the quantum set \cite{Popescu2006,Brunner2009,Pawlowski2009,Navascues2010,Oppenheim2010}. The principle of information causality has already proved useful to derive bounds on the quantum set in non-trivial scenarios \cite{Cavalcanti2010}. In this paper, we use the principle of \textit{macroscopic locality} \cite{Navascues2010} to derive quantum Bell inequalities. The recipe is applicable to all bipartite scenarios. Focusing on $nm22$ scenarios, we prove that a quantum Bell inequality can be obtained from any standard Bell inequality through (\ref{replace}). Such inequalities are the only analytical approximations to the quantum set available to date.

\section{Tools}

\subsection{Macroscopic locality}
Before presenting our results, we review briefly the idea of macroscopic locality (ML)  \cite{Navascues2010}. We consider a bipartite system, with a source emitting pairs of particles to Alice and Bob. Alice can choose a measurement $x\in X=\{1\ldots,m_A\}$ from a set of $m_A$ possible settings, each producing $d_A$ possible outcomes, denoted as $a\in \mathcal{A}=\{1,\ldots,d_A\}$. Similarly, Bob can choose a measurement $y$ from the set $Y=\{1,\ldots,m_B\}$, with $d_B$ outcomes $b\in \mathcal{B}=\{1,\ldots,d_B\}$. One calls a \textit{microscopic experiment} one in which single-particle detections are possible. In this case, routine sampling leads to reconstructing $P(a,b|x,y)$. By contrast, a \textit{macroscopic experiment} is defined as follows: in each run, $N$ pairs are sent; Alice measures all her particles with setting $x$ and records the number $n_a$ of particles that produced the outcome $a\in\mathcal{A}$; Bob does similarly. After repeating this procedure several times, one can reconstruct $P(\vec{n}_A,\vec{n}_B|x,y)$, where $\vec{n}_A=[n_{a=1},...,n_{a=d_A}]$ and $\vec{n}_B=[n_{b=1},...,n_{b=d_B}]$. Clearly, some microscopic information is lost as soon as $N>1$: Alice and Bob do not know which of their outcomes came from the same pair.

The principle of ML demands that, in the limit $N\rightarrow\infty$ and under the assumption that Alice's and Bob's measurement devices cannot detect particle number fluctuations smaller than $O(\sqrt{N})$, the macroscopic statistics should not violate any Bell inequality. The principle is not trivial: for instance, one can easily check that, if the source were a Popescu-Rohrlich (PR) box, the macroscopic statistics would violate CHSH up to 4, just as the microscopic statistic do. In fact, ML is respected if and only if the microscopic statistics $P(a,b|x,y)$ belong to $Q^1$ \cite{Navascues2010}.

\subsection{Recipe for quantum Bell inequalities}
It is now easy to formulate our recipe to generate quantum Bell inequalities. Consider a data processing of the macroscopic variables
\begin{align}
\vec{n}_A \rightarrow \alpha\in\mathcal{A}' &,& 	\vec{n}_B \rightarrow \beta\in\mathcal{B}'\,.\label{dataproc}
\end{align} This defines a scenario with $m_A$ measurements for Alice, $m_B$ for Bob, $|\mathcal{A}'|=d_A'$ outcomes for Alice and $|\mathcal{B}'|=d_B'$ for Bob -- in short, denoted as $m_Am_Bd_A'd_B'$ scenario. From ML it follows that, if $P(a,b|x,y)$ belongs to the quantum set, $P(\alpha,\beta|x,y)$ cannot violate any Bell inequality. In other words, for any Bell inequality ${\cal I}\leq 0$ in the scenario $m_Am_Bd_A'd_B'$, the condition ${\cal I}[P(\alpha,\beta|x,y)]\leq 0$ defines a quantum Bell inequality for the microscopic correlations. If in addition $P(\alpha,\beta|x,y)$ can be expressed as a function of the parameters of $P(a,b|x,y)$, an analytical bound on the quantum set is obtained.

Of course, we know already that none of these criteria will come closer to the quantum set than $Q^1$ does; and in general, the bound may even be less tight than $Q^1$ (we show an explicit example below). However, to date, this is the only known systematic approach to obtain analytical approximations to the quantum set in any scenario. In what follows, we study some specific examples of quantum Bell inequalities generated with this recipe.

\section{Case studies}

\subsection{Two outcomes: Sign-binning}
Let $d_A=d_B=2$ and for convenience we use $\mathcal{A}=\mathcal{B}=\{-1,1\}$. Consider first a specific choice of measurements $x,y$ (so, for simplicity, we omit them in the notation). From her macroscopic statistics, Alice can infer the microscopic marginal average
\begin{align}
\ave{a}\,=\,\ave{\Delta n_A}/N& \;\;\textrm{with}& \Delta n_A\equiv n_{a=+1}-n_{a=-1}\,.
\end{align} A natural data processing that gives back binary outcomes ($\mathcal{A}'=\{-1,1\}$) consists in comparing $\Delta n_A$ with $\ave{a}$ in each run:
\begin{align}
\textrm{Sign-binning:}&\;\;\alpha\,=&\left\{\begin{array}{ll} +1 & \textrm{if $\Delta n_A\geq \ave{a}$}\\ -1 & \textrm{otherwise} \end{array}\right.\,.\label{signbinning}
\end{align} Bob can of course do the same. Now we want to relate $\alpha$ and $\beta$ obtained by sign-binning to the properties of the microscopic distribution. Any macroscopic run uses $N$ pairs: denote therefore $\{a^{(1)},...,a^{(N)}\}$ the results Alice would have observed on those $N$ pairs in a microscopic experiment, and similarly $\{b^{(1)},...,b^{(N)}\}$ for Bob. Then we define
\begin{align}
	\alpha &= \sgn \left( \sum_{k=1}^N \frac{a^{(k)} -\ave{a}}{\sqrt{N}}\right) \equiv \sgn(a'), \\
	\beta &= \sgn \left( \sum_{k=1}^N \frac{b^{(k)} -\ave{b}}{\sqrt{N}} \right) \equiv \sgn(b') .
\end{align}
In the limit $N\rightarrow\infty$, due to the central limit theorem, the joint probability distribution of the variables $(a',b')$ is a gaussian distribution $G_\Gamma(a',b')$ with zero mean and covariance matrix, $\Gamma$ determined by the microscopic distribution \cite{Navascues2010}
\begin{align}	\Gamma=\matrixww{1-\ave{a}^2}{\ave{ab}-\ave{a}\ave{b}}{\ave{ab}-\ave{a}\ave{b}}{1-\ave{b}^2}\,.
\end{align} Then, with the notation $\overline{f}=\int da'db' \; G_\Gamma(a',b')f(a',b')$, a straightforward integration shows that
\ba
\overline{\alpha}=\overline{\beta}&=&0, \;\;\; \overline{\alpha\beta} =\frac{2}{\pi}\arcsin D_{xy} \,. \label{eqn:avealpha} 
\ea
All that we have done is valid for any pair of measurements $x,y$. Therefore, as claimed, the replacement (\ref{replace}) is not specific to the CHSH scenario: given \textit{any} bipartite Bell inequality with binary outputs, one can impose that $P(\alpha,\beta|x,y)$ respects it and then replace all the $\overline{\alpha_x}$, $\overline{\beta_x}$ and $\overline{\alpha_x\beta_y}$ with the corresponding expressions above: the resulting inequality will automatically define a quantum Bell inequality for the microscopic distribution $P(a,b|x,y)$. Let us denote the set of correlations satisfying these inequalities as $Q^{SB}$. Of course, from the above discussion it follows that $Q^1\subset Q^{SB}$.

Many bipartite Bell inequalities with binary output are known \cite{Pal2008,Brunner2008}, some of which have intriguing properties (e.g., they are maximally violated by non-maximally entangled states, or define dimension witnesses). For the sake of this paper, we consider only two examples in what follows.

\subsection{Sign binning for $2n22$} With CHSH being a special case $(n=2)$ of this scenario, we note that the only relevant Bell inequality in $2n22$, up to symmetries, is the CHSH inequality \cite{Collins2004}. Using the replacement in (\ref{eqn:avealpha}), we obtain the analytical inequality 
\begin{align}
	\left| \arcsin D_{1i} + \arcsin D_{2i} + \arcsin D_{1j} - \arcsin D_{2j}     \right| \leq \pi \label{nchsh}
\end{align}
for all $i,j=1,\ldots,n$ as a quantum Bell inequality. For the 2222 scenario, we recover the inequality given in \cite{Navascues2010}, thus proving that $Q^{SB}= Q^1$ in that case. In fact, $Q^1=Q^{SB}$ also holds in any $2n22$ scenario: this is not unexpected, since it is known that CHSH is the only relevant inequality for those scenarios \cite{Collins2004}. We give an explicit proof in the Appendix \ref{append:2n22}. In summary, in $2n22$ scenarios, the distribution $P(\vec{n}_A,\vec{n}_B|x,y)$ is local if and only if $P(\alpha,\beta|x,y)$ is local. This remarkable coincidence is however not general as shown in the next paragraph.

\subsection{$3322$ Scenario: sign-binning and beyond}
For $3322$, besides CHSH,  there is only one inequivalent Bell inequality denoted as $I_{3322}$ \cite{Froissart1981,Collins2004}. The set $Q^{SB}$ is therefore defined by both (\ref{nchsh}) and 
\begin{align}
	\left|\sum_{x+y\leq 4} \arcsin D_{xy} - \arcsin D_{32} - \arcsin D_{23} \right| \leq 2\pi. \label{3322condition1}
\end{align}
By considering a specific slice of the 3322 no-signaling polytope, we found that this set of inequalities gives weaker constraints than ML, i.e. $Q^1 \subset Q^{SB}$. This can be seen from \referfig{slice3322}. In the two dimensional slice defined by the three extremal points, $P_1$, $P_2$ and $P_n$, $Q^{SB}$ gives a bound of $I_{3322}=0.4$ while ML (i.e. $Q^1$) gives 0.2.

\begin{figure}[htbp!]
	\includegraphics[scale=0.15]{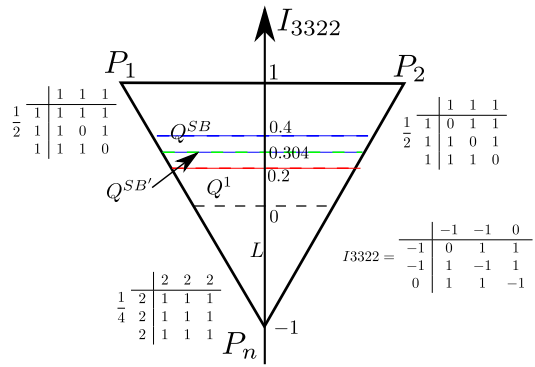}
	\caption{A slice of the 3322 no-signaling polytope defined by the two extremal points $P_1$, $P_2$ ($I_{3322}=1$) and $P_n$ (random noise).  In this slice, the bound for the local set $L$ and $Q$ are the same, which is $0$. The notation of the points is taken from \cite{Collins2004} and the extremality of $P_1$ and $P_2$ follows from \cite{Jones2005,Barrett2005}.\label{slice3322}}
\end{figure}

It is not astonishing that the quantum Bell inequalities obtained from sign binning may be weaker. Geometrically, sign binning is obtained by slicing the Gaussian curve into two symmetrical parts, each part representing a macroscopic outcome. In general though, macroscopic locality is recovered by slicing the curve into infinitely many parts or macroscopic outcomes \cite{Navascues2010}, as shown in \referfig{gaussian2222}.

This motivates us to investigate different forms of binning. Here we show a form of data processing induced by the Gaussian function 
\begin{align}
	P(x)=\exp\Big(-\frac{x^2}{2\sigma}\Big). \label{gaussiansigma}
\end{align}
The binning is then defined as follows,
\begin{align}
\textrm{3-binning:}\;\; &\alpha\,=&\left\{
\begin{array}{ll} 
	0 & \textrm{with probability $P(\delta n)$}\\
	+1 & \textrm{otherwise and $\delta n>0$}\\
	-1 & \textrm{otherwise} 
\end{array}\right.\,.
\end{align}
where $\delta n\equiv \Delta n_A-\ave{a}$. This generates a new set of probabilities $\{P(\alpha',\beta'|x,y)\}$ in the $3333$-scenario, all of which can be computed analytically. Thus, this allows us, in principle, to have an analytical quantum Bell inequality. Here, we numerically determine the bound and find that the new set, $Q^{SB'}$ gives a better approximation to $Q$. This can be seen in \referfig{slice3322} where $Q^{SB'}$ gives 0.304 by taking $\sigma=0.028$ in Eq. (\ref{gaussiansigma}).

\subsection{Binning for $n>2$ microscopic outcomes.}

The sign binning processing is of course a natural processing for binary outcome settings, though not optimal. It can also be applied to settings with more than two outcomes; but here we sketch a more natural processing for such scenarios. As an example, we will focus on the case of $2233$: $X=Y=\{0,1 \} $ and $\mathcal{A}=\mathcal{B}=\{0,1,2\}$. We first define the quantity 
\begin{align}
	I^A_i = \dfrac{n_i(\vec{a}) - p^A_iN}{\sqrt{N}}, \forall i\in \mathcal{A} 
\end{align}
for a particular fixed value of $x$ and $y$. The data processing, in this case called the triangle-binning, is then defined as
\begin{align}
	\alpha\,=k, \hspace{0.4cm} \textrm{s.t.}\hspace{0.4cm} I^A_k = \max_{i\in\mathcal{A}}\{I^A_i\} .
\end{align}
Bob can similarly do the same. In the limit $N\rightarrow \infty$, the probability distribution for the variables $(\alpha_x,\beta_y)$ is a Gaussian distribution $G_\Gamma(\alpha_x,\beta_y)$ with the covariance matrix $\Gamma$ defined in the usual sense. The formulas for the probabilities $P(\alpha,\beta|x,y)$ include definite trigonometric integrals; they may not lead to elegant analytical expressions, but they can then be evaluated numerically to check the method. It is known that $2233$ has only one relevant Bell inequality of the Collins-Gisin-Linden-Massar-Popescu (CGLMP) type \cite{Collins2002}: $I_{CGLMP}\leq 0$. We can then impose the condition $I_{CGLMP}[P(\alpha,\beta|x,y)]\leq 0$. The result is shown in \referfig{line2233} where the constrained set of correlation is now denoted as $Q^{TB}$. Contrary to sign binning, triangle binning does not reproduce the predictions of ML. Triangle binning may thus not be worth pursuing further, but it shows that generalization is indeed possible.
\begin{figure}[htbp!]
	\includegraphics[scale=0.23]{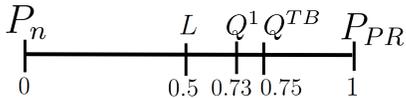}
	\caption{The line of random marginals defined by the random noise, $P_n$ and the generalized PR box, $P_{PR}$ \cite{Scarani2006}. The bounds given by the local set, $L$, macroscopic locality, $Q^1$ and triangle binning, $Q^{\mathrm{TB}}$ are 0.5, 0.75 and 0.73 respectively.\label{line2233}}
\end{figure}

\begin{figure}[htbp!]
		\hspace{-2.8cm}
	\includegraphics[scale=0.16]{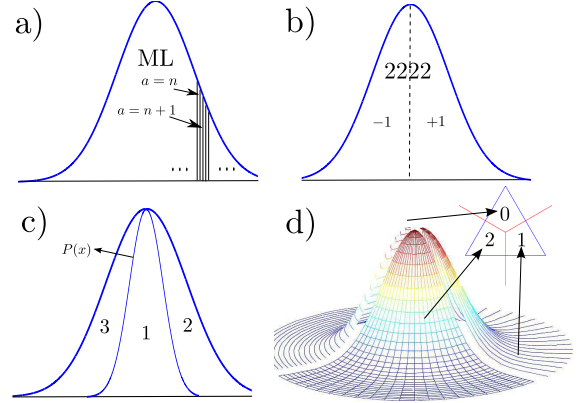}
	\caption{a) ML in the sense of binning into infinitely many parts. b) Sign binning for scenarios with binary outcomes. c) 3-binning induced by the probability, $P(x)$. d) Triangle-binning for scenarios with three outcomes. \label{gaussian2222}}s
\end{figure}

\section{Conclusion}
In this paper, we have studied a new post-processing method, called sign binning, to define a set of quantum correlations denoted as $Q^{SB}$, which necessarily contains the quantum set, $Q$. Sign binning, as mentioned, cannot give a bound on the quantum set that is tighter than $Q^1$, but it does give a straightforward, intuitive way of generating quantum Bell inequalities. To our knowledge, this recipe is the only method to date to derive analytical quantum Bell inequalities.

\begin{acknowledgments}
This work was supported by the National Research Foundation and the Ministry of Education, Singapore, and by the European Project QUEVADIS.
\end{acknowledgments}




\appendix

\section{Locality of Sign Binning is Equivalent to Macroscopic Locality in $2n22$ \label{append:2n22}}

\begin{lemma}
\label{equiva}
Let $\Gamma$ be an $n+2$ square matrix of the form
\begin{align}
\Gamma=\left(\begin{array}{cc}A&C\\C^T&B\end{array}\right),
\label{structure}
\end{align}
\noindent where $C$ is a given $2\times n$ real matrix and $A,B$ are such that $A_{ii}=B_{jj}=1$ for $i=1,2$ and $j=1,...,n$. Then there exists a choice of the remaining entries of $A$ and $B$ such that $\Gamma\geq 0$ iff there exists $x\in [-1,1]$ such that
\begin{align}
1-x^2-C_{1i}^2-C_{2i}^2+2xC_{1i}C_{2i}\geq 0   ,
\label{condition}
\end{align}
\noindent for $i=1,...,n$.
\end{lemma}

\begin{proof}
Let us first prove that, if condition (\ref{condition}) holds, then $\Gamma$ can be made positive semidefinite. Suppose that, indeed, such an $x$ exists and $|x|<1$. Then, we can take $A$ to be
\begin{align}
A=\left(\begin{array}{cc}1&x\\x&1\end{array}\right)>0.
\end{align}
According to Schur's theorem \cite{Horn1999}, if $A>0$, a matrix of the form (\ref{structure}) is positive semidefinite iff $B'\equiv B-C^TA^{-1}C\geq 0$. Since the non-diagonal entries of $B$ are not determined a priori, we can always choose them such that $B'_{ij}=0$ for $i\not= j$. To see that $B'$ is positive semidefinite, we then only have to show that $B'_{ii}\geq 0$. But
\begin{align}
B'_{ii}=\frac{1-x^2-C_{1i}^2-C_{2i}^2+2xC_{1i}C_{2i}}{1-x^2},
\end{align}
\noindent that is non-negative by hypothesis. We have just proven that, for $|x|<1$, condition (\ref{condition}) grants positive semidefiniteness. Suppose now that (\ref{condition}) holds for $x=1$. Then the equation reads
\begin{align}
-(C_{1i}-C_{2i})^2\geq 0, \mbox{ for } i=1,...,n.
\end{align}
\noindent It follows that $C_{1i}=C_{2i}$ for all $i$. In order to show that $\Gamma$ can be completed to a positive semidefinite matrix, take the orthonormal basis $\{\ket{0},\ket{1}\}$ and define the vectors:
\begin{align}
\vec{v}_{1,2}\equiv \ket{0};\vec{v}_{i+2}\equiv C_{1i}\ket{0}+\sqrt{1-C^2_{1i}}\ket{1}.
\end{align}
Then, the Gram matrix $\Gamma'_{ij}=\vec{v}_i\cdot \vec{v}_j$ is positive semidefinite, has 1s in the diagonal and its off-diagonal submatrix coincides with $C$.

\noindent The case $x=-1$ can be treated analogously (simply take $\vec{v}_1=-\vec{v}_2$).

Now we will prove the opposite implication: suppose that there is some way to complete $\Gamma$ such that $\Gamma\geq 0$. Let $\tilde{\Gamma}$ be such completion and take $x= \tilde{A}_{12}$. If $x=\pm 1$, then the Gram decomposition of $\tilde{\Gamma}_{ij}=\vec{v}_i\cdot \vec{v}_j$ \cite{Horn1999} is such that $\|\vec{v}_1\|=\|\vec{v}_2\|=1$ and $\vec{v}_1\cdot \vec{v}_2=\pm 1$. This implies that $\vec{v}_1=\pm \vec{v}_2$, and so $C_{1i}=\vec{v}_1\cdot\vec{v}_{2+i}=\pm\vec{v}_2\cdot\vec{v}_{2+i}=\pm C_{2i}$, and condition (\ref{condition}) holds for $x=\pm 1$.

Suppose that, on the contrary, $|x|<1$. Then $A>0$, so, by Schur's theorem, $\tilde{B}'_{ii}\geq 0$, and condition (\ref{condition}) holds.

\end{proof}

\begin{theorem}
	\label{theoremequiv}
Let $\Gamma$ be a matrix such as the one appearing in the definition of the previous lemma. Then, $\Gamma$ can be made positive semidefinite iff, for all $i,j=1,...,n$, $i\not=j$,
\begin{align}
\Big| \arcsin(C_{1i})+\arcsin(C_{2i})+  \arcsin(C_{1j})-  \nonumber\\
 \arcsin(C_{2j}) \Big| \leq \pi,
\label{arcsines}
\end{align}
\noindent plus permutations of the minus sign.

\end{theorem}

\begin{proof}
By lemma \ref{equiva}, positive semidefiniteness is equivalent to the existence of an $x\in [-1,1]$ satisfying the conditions (\ref{condition}). Without loss of generality, we assume that $C_{1i}\equiv \sin(\phi_i),C_{2i}\equiv \sin(\theta_i)$, for $-\pi/2\leq\theta_i,\phi_i\leq \pi/2$. Then, conditions (\ref{condition}) can be reexpressed as
\begin{align}
-\cos(\phi_i+\theta_i)\leq x\leq \cos(\phi_i-\theta_i).
\end{align}

An $x$ satisfying all these conditions exists iff the minimum of the upper limits is greater than or equal to the maximum of the lower limits. In other words, $\Gamma$ can be completed iff
\begin{align}
-\cos(\phi_j+\theta_j)\leq \cos(\phi_i-\theta_i), \forall i,j.
\end{align}

Call $\alpha_i\equiv |\phi_i-\theta_i|$, $\beta_j\equiv |\phi_j+\theta_j|$. Then, $0\leq\alpha_i,\beta_j\leq \pi$, and the positivity condition reads
\begin{align}
\cos(\alpha_i)+\cos(\beta_j)\geq 0.
\end{align}

Running through all possibilities ($[\alpha_i\leq \pi/2,\beta_j\leq \pi/2],[\alpha_i\leq \pi/2,\beta_j\geq \pi/2],[\alpha_i\geq \pi/2,\beta_j\leq \pi/2],[\alpha_i\geq \pi/2,\beta_j\geq \pi/2]$), one can check that this condition is equivalent to $\alpha_i+\beta_j\leq \pi$, and so we arrive at equations (\ref{arcsines}).

\end{proof}

We will now prove the claimed result in the article that $Q^{SB}=Q^1$ for $2n22$. This proof is an extension of the proof presented in \cite{Navascues2008}. For the case of $2n22$, let $\Gamma^1$ be a certificate of order $1$ for a particular $P(a,b|x,y)$. Then we have
\begin{align}
	\Gamma^1 = \left( 
		\begin{array}{ccc}
			1 & C_A & C_B \\
			C_A^T & X & Z \\
			C_B^T & Z^T & Y
		\end{array}
	\right),
\end{align}
where  $X_{ii}=Y_{jj}=1$ for $i=1,2$, $j=1,...,n$. Also, $C_A = (C_{A1},C_{A2})$ and $C_B= (C_{B1},\ldots,C_{Bn})$ are the marginal correlations for Alice's and Bob's measurement respectively. By Schur's Lemma \cite{Horn1999}, $\Gamma^1 \geq 0$ is equivalent to the positive semidefiniteness of 
\begin{align}
	\overline{\Gamma}^1 &= \left( \begin{array}{cc} X & Z \\ Z^T & Y \end{array} \right) - 
		\left( \begin{array} {c} C_A^T \\ C_B^T \end{array}\right) \left( C_A,C_B\right)  
\end{align}
where the matrix $\overline{\Gamma}^1$ has diagonal elements $\{1-C^2_{A1},1-C^2_{A2},1-C^2_{B1},\ldots,1-C^2_{Bn}\}$. We may assume that all the diagonal elements are non zero for $n\geq 2$; if any of them is zero, then the outcome of that particular measurement in deterministic and can be accounted for with a local hidden variable model. Therefore we can multiply $\overline{\Gamma}^1$ on both sides with the diagonal matrix, $M=\{\sqrt{1-C^2_{A1}},\sqrt{1-C^2_{A2}},\sqrt{1-C^2_{B1}},\ldots,\sqrt{1-C^2_{Bn}}\}$. The condition $\overline{\Gamma}^1\geq 0$ is then equivalent to
\begin{align}
	\Gamma=\left(\begin{array}{cc}A&C\\C^T&B\end{array}\right)\geq 0,
\end{align}
where the $2\times n$ matrix $C$ has elements $C_{ij} = (C_{ij}-C_iC_j) / \sqrt{1-C_i^2}\sqrt{1-C_j^2}$ and $A_{ii}=B_{jj}=1$ for $i=1,2$ and $j=1,...,n$. Now by Lemma \ref{equiva} and Theorem \ref{theoremequiv}, a certificate of order 1 exists, and therefore probability distribution, $P(a,b|x,y)$ is $Q^1$ if and only if condition in (\ref{nchsh}) is satisfied.



\bibliography{mycollection}

\begin{thebibliography}{24}
\expandafter\ifx\csname natexlab\endcsname\relax\def\natexlab#1{#1}\fi
\expandafter\ifx\csname bibnamefont\endcsname\relax
  \def\bibnamefont#1{#1}\fi
\expandafter\ifx\csname bibfnamefont\endcsname\relax
  \def\bibfnamefont#1{#1}\fi
\expandafter\ifx\csname citenamefont\endcsname\relax
  \def\citenamefont#1{#1}\fi
\expandafter\ifx\csname url\endcsname\relax
  \def\url#1{\texttt{#1}}\fi
\expandafter\ifx\csname urlprefix\endcsname\relax\def\urlprefix{URL }\fi
\providecommand{\bibinfo}[2]{#2}
\providecommand{\eprint}[2][]{\url{#2}}

\bibitem[{\citenamefont{Bell}(1966)}]{Bell1966}
\bibinfo{author}{\bibfnamefont{J.~S.} \bibnamefont{Bell}},
  \bibinfo{journal}{Rev. Mod. Phys.} \textbf{\bibinfo{volume}{38}},
  \bibinfo{pages}{447} (\bibinfo{year}{1966}).

\bibitem[{\citenamefont{Popescu and Rohrlich}(1994)}]{Popescu1994}
\bibinfo{author}{\bibfnamefont{S.}~\bibnamefont{Popescu}} \bibnamefont{and}
  \bibinfo{author}{\bibfnamefont{D.}~\bibnamefont{Rohrlich}},
  \bibinfo{journal}{Found. Phys.} \textbf{\bibinfo{volume}{24}},
  \bibinfo{pages}{379} (\bibinfo{year}{1994}).

\bibitem[{\citenamefont{Cirel'son}(1980)}]{Cirel'son1980}
\bibinfo{author}{\bibfnamefont{B.~S.} \bibnamefont{Cirel'son}},
  \bibinfo{journal}{Lett. Math. Phys.} \textbf{\bibinfo{volume}{4}},
  \bibinfo{pages}{93} (\bibinfo{year}{1980}).

\bibitem[{\citenamefont{Clauser et~al.}(1969)\citenamefont{Clauser, Horne,
  Shimony, and Holt}}]{Clauser1969}
\bibinfo{author}{\bibfnamefont{J.~F.} \bibnamefont{Clauser}},
  \bibinfo{author}{\bibfnamefont{M.~A.} \bibnamefont{Horne}},
  \bibinfo{author}{\bibfnamefont{A.}~\bibnamefont{Shimony}}, \bibnamefont{and}
  \bibinfo{author}{\bibfnamefont{R.~A.} \bibnamefont{Holt}},
  \bibinfo{journal}{Phys. Rev. Lett.} \textbf{\bibinfo{volume}{23}},
  \bibinfo{pages}{880} (\bibinfo{year}{1969}).

\bibitem[{\citenamefont{Tsirelson}(1987)}]{Tsirelson1987}
\bibinfo{author}{\bibfnamefont{B.}~\bibnamefont{Tsirelson}},
  \bibinfo{journal}{J. Sov. Math.} \textbf{\bibinfo{volume}{36}},
  \bibinfo{pages}{557} (\bibinfo{year}{1987}).

\bibitem[{\citenamefont{Landau}(1988)}]{Landau1988}
\bibinfo{author}{\bibfnamefont{L.}~\bibnamefont{Landau}},
  \bibinfo{journal}{Found. Phys.} \textbf{\bibinfo{volume}{18}},
  \bibinfo{pages}{449} (\bibinfo{year}{1988}).

\bibitem[{\citenamefont{Masanes}(2003)}]{Masanes2003}
\bibinfo{author}{\bibfnamefont{L.}~\bibnamefont{Masanes}},
  \bibinfo{journal}{Arxiv preprint quant-ph/0309137}  (\bibinfo{year}{2003}).

\bibitem[{\citenamefont{Navascues et~al.}(2007)\citenamefont{Navascues,
  Pironio, and Ac\'\i{}n}}]{Navascues2007}
\bibinfo{author}{\bibfnamefont{M.}~\bibnamefont{Navascues}},
  \bibinfo{author}{\bibfnamefont{S.}~\bibnamefont{Pironio}}, \bibnamefont{and}
  \bibinfo{author}{\bibfnamefont{A.}~\bibnamefont{Ac\'\i{}n}},
  \bibinfo{journal}{Phys. Rev. Lett.} \textbf{\bibinfo{volume}{98}},
  \bibinfo{pages}{010401} (\bibinfo{year}{2007}).

\bibitem[{\citenamefont{Navascues et~al.}(2008)\citenamefont{Navascues,
  Pironio, and Ac\'in}}]{Navascues2008}
\bibinfo{author}{\bibfnamefont{M.}~\bibnamefont{Navascues}},
  \bibinfo{author}{\bibfnamefont{S.}~\bibnamefont{Pironio}}, \bibnamefont{and}
  \bibinfo{author}{\bibfnamefont{A.}~\bibnamefont{Ac\'in}},
  \bibinfo{journal}{New J. Phys.} \textbf{\bibinfo{volume}{10}},
  \bibinfo{pages}{073013} (\bibinfo{year}{2008}).

\bibitem[{\citenamefont{Popescu}(2006)}]{Popescu2006}
\bibinfo{author}{\bibfnamefont{S.}~\bibnamefont{Popescu}},
  \bibinfo{journal}{Nat. Phys.} \textbf{\bibinfo{volume}{2}},
  \bibinfo{pages}{507} (\bibinfo{year}{2006}).

\bibitem[{\citenamefont{Brunner and Skrzypczyk}(2009)}]{Brunner2009}
\bibinfo{author}{\bibfnamefont{N.}~\bibnamefont{Brunner}} \bibnamefont{and}
  \bibinfo{author}{\bibfnamefont{P.}~\bibnamefont{Skrzypczyk}},
  \bibinfo{journal}{Phys. Rev. Lett.} \textbf{\bibinfo{volume}{102}},
  \bibinfo{pages}{160403} (\bibinfo{year}{2009}).

\bibitem[{\citenamefont{Pawlowski et~al.}(2009)\citenamefont{Pawlowski,
  Paterek, Kaszlikowski, Scarani, Winter, and Zukowski}}]{Pawlowski2009}
\bibinfo{author}{\bibfnamefont{M.}~\bibnamefont{Pawlowski}},
  \bibinfo{author}{\bibfnamefont{T.}~\bibnamefont{Paterek}},
  \bibinfo{author}{\bibfnamefont{D.}~\bibnamefont{Kaszlikowski}},
  \bibinfo{author}{\bibfnamefont{V.}~\bibnamefont{Scarani}},
  \bibinfo{author}{\bibfnamefont{A.}~\bibnamefont{Winter}}, \bibnamefont{and}
  \bibinfo{author}{\bibfnamefont{M.}~\bibnamefont{Zukowski}},
  \bibinfo{journal}{Nature} \textbf{\bibinfo{volume}{461}},
  \bibinfo{pages}{1101} (\bibinfo{year}{2009}).

\bibitem[{\citenamefont{Navascues and Wunderlich}(2010)}]{Navascues2010}
\bibinfo{author}{\bibfnamefont{M.}~\bibnamefont{Navascues}} \bibnamefont{and}
  \bibinfo{author}{\bibfnamefont{H.}~\bibnamefont{Wunderlich}},
  \bibinfo{journal}{Proc. R. Soc. A} \textbf{\bibinfo{volume}{466}},
  \bibinfo{pages}{881} (\bibinfo{year}{2010}).

\bibitem[{\citenamefont{Oppenheim and Wehner}(2010)}]{Oppenheim2010}
\bibinfo{author}{\bibfnamefont{J.}~\bibnamefont{Oppenheim}} \bibnamefont{and}
  \bibinfo{author}{\bibfnamefont{S.}~\bibnamefont{Wehner}},
  \bibinfo{journal}{Arxiv preprint arXiv:1004.2507}  (\bibinfo{year}{2010}).

\bibitem[{\citenamefont{Cavalcanti et~al.}(2010)\citenamefont{Cavalcanti,
  Salles, and Scarani}}]{Cavalcanti2010}
\bibinfo{author}{\bibfnamefont{D.}~\bibnamefont{Cavalcanti}},
  \bibinfo{author}{\bibfnamefont{A.}~\bibnamefont{Salles}}, \bibnamefont{and}
  \bibinfo{author}{\bibfnamefont{V.}~\bibnamefont{Scarani}},
  \bibinfo{journal}{Arxiv preprint arXiv:1008.2624}  (\bibinfo{year}{2010}).

\bibitem[{\citenamefont{Pal and Vertesi}(2008)}]{Pal2008}
\bibinfo{author}{\bibfnamefont{K.~F.} \bibnamefont{Pal}} \bibnamefont{and}
  \bibinfo{author}{\bibfnamefont{T.}~\bibnamefont{Vertesi}},
  \bibinfo{journal}{Phys. Rev. A} \textbf{\bibinfo{volume}{77}},
  \bibinfo{pages}{042105} (\bibinfo{year}{2008}).

\bibitem[{\citenamefont{Brunner and Gisin}(2008)}]{Brunner2008}
\bibinfo{author}{\bibfnamefont{N.}~\bibnamefont{Brunner}} \bibnamefont{and}
  \bibinfo{author}{\bibfnamefont{N.}~\bibnamefont{Gisin}},
  \bibinfo{journal}{Phys. Lett. A} \textbf{\bibinfo{volume}{372}},
  \bibinfo{pages}{3162 } (\bibinfo{year}{2008}), ISSN
  \bibinfo{issn}{0375-9601}.

\bibitem[{\citenamefont{Collins and Gisin}(2004)}]{Collins2004}
\bibinfo{author}{\bibfnamefont{D.}~\bibnamefont{Collins}} \bibnamefont{and}
  \bibinfo{author}{\bibfnamefont{N.}~\bibnamefont{Gisin}}, \bibinfo{journal}{J.
  Phys. A: Math. Gen.} \textbf{\bibinfo{volume}{37}}, \bibinfo{pages}{1775}
  (\bibinfo{year}{2004}).

\bibitem[{\citenamefont{Froissart}(1981)}]{Froissart1981}
\bibinfo{author}{\bibfnamefont{M.}~\bibnamefont{Froissart}},
  \bibinfo{journal}{Il Nuovo Cimento B (1971-1996)}
  \textbf{\bibinfo{volume}{64}}, \bibinfo{pages}{241} (\bibinfo{year}{1981}).

\bibitem[{\citenamefont{Jones and Masanes}(2005)}]{Jones2005}
\bibinfo{author}{\bibfnamefont{N.~S.} \bibnamefont{Jones}} \bibnamefont{and}
  \bibinfo{author}{\bibfnamefont{L.}~\bibnamefont{Masanes}},
  \bibinfo{journal}{Phys. Rev. A} \textbf{\bibinfo{volume}{72}},
  \bibinfo{pages}{052312} (\bibinfo{year}{2005}).

\bibitem[{\citenamefont{Barrett and Pironio}(2005)}]{Barrett2005}
\bibinfo{author}{\bibfnamefont{J.}~\bibnamefont{Barrett}} \bibnamefont{and}
  \bibinfo{author}{\bibfnamefont{S.}~\bibnamefont{Pironio}},
  \bibinfo{journal}{Phys. Rev. Lett.} \textbf{\bibinfo{volume}{95}},
  \bibinfo{pages}{140401} (\bibinfo{year}{2005}).

\bibitem[{\citenamefont{Collins et~al.}(2002)\citenamefont{Collins, Gisin,
  Linden, Massar, and Popescu}}]{Collins2002}
\bibinfo{author}{\bibfnamefont{D.}~\bibnamefont{Collins}},
  \bibinfo{author}{\bibfnamefont{N.}~\bibnamefont{Gisin}},
  \bibinfo{author}{\bibfnamefont{N.}~\bibnamefont{Linden}},
  \bibinfo{author}{\bibfnamefont{S.}~\bibnamefont{Massar}}, \bibnamefont{and}
  \bibinfo{author}{\bibfnamefont{S.}~\bibnamefont{Popescu}},
  \bibinfo{journal}{Phys. Rev. Lett.} \textbf{\bibinfo{volume}{88}},
  \bibinfo{pages}{040404} (\bibinfo{year}{2002}).

\bibitem[{\citenamefont{Scarani et~al.}(2006)\citenamefont{Scarani, Gisin,
  Brunner, Masanes, Pino, and Ac\'\i{}n}}]{Scarani2006}
\bibinfo{author}{\bibfnamefont{V.}~\bibnamefont{Scarani}},
  \bibinfo{author}{\bibfnamefont{N.}~\bibnamefont{Gisin}},
  \bibinfo{author}{\bibfnamefont{N.}~\bibnamefont{Brunner}},
  \bibinfo{author}{\bibfnamefont{L.}~\bibnamefont{Masanes}},
  \bibinfo{author}{\bibfnamefont{S.}~\bibnamefont{Pino}}, \bibnamefont{and}
  \bibinfo{author}{\bibfnamefont{A.}~\bibnamefont{Ac\'\i{}n}},
  \bibinfo{journal}{Phys. Rev. A} \textbf{\bibinfo{volume}{74}},
  \bibinfo{pages}{042339} (\bibinfo{year}{2006}).

\bibitem[{\citenamefont{Horn and Johnson}(1999)}]{Horn1999}
\bibinfo{author}{\bibfnamefont{R.~A.} \bibnamefont{Horn}} \bibnamefont{and}
  \bibinfo{author}{\bibfnamefont{C.~R.} \bibnamefont{Johnson}},
  \emph{\bibinfo{title}{Matrix Analysis}} (\bibinfo{publisher}{Cambridge
  University Press}, \bibinfo{year}{1999}).

\end{thebibliography}
\end{document}